\documentclass[3p,a4paper,times,final]{article}

\usepackage[draft]{optional}

\usepackage[utf8]{inputenc}
\usepackage[english] {babel} % Mise en page française

\usepackage[T1]{fontenc}      % Encodage des caractères
\usepackage{amsmath,amsfonts,amssymb,amsthm}   % Formules de math
\usepackage{textcomp}         % Symboles euro, degré, etc...
\usepackage{float}
\usepackage{standalone}
\usepackage{stmaryrd}
\usepackage{tikz,ifthen}
\usetikzlibrary{decorations.pathmorphing,decorations.pathreplacing,arrows,automata,patterns,positioning}
\usepackage{latexsym}
 \usepackage{graphicx} % Images (eps non supporté)

\usepackage{xcomment}
\usepackage{xspace}
\usepackage{url}
\newtheorem{theorem}{Theorem}[section] % Numéroté à partir des sections
\newtheorem{lemma}[theorem]{Lemma} % Numéroté comme les thm
\newtheorem{corollary}[theorem]{Corollary} % Numéroté comme les thm
\newtheorem{definition}[theorem]{Definition}  % Numéroté comme les thm
 \colorlet{colorCarreUpRight}{green!30!black}
 \colorlet{colorCarreUpLeft}{green!60!black}
 \colorlet{colorCarreDownRight}{yellow!50!black}
 \colorlet{colorCarreDownLeft}{yellow!90!black}
 \colorlet{colorCarreBorder}{black}
 \colorlet{colorTroisQuart}{white}
 \colorlet{colorUnQuartBas}{red!30}
 \colorlet{colorSignalCarres}{red}
 \colorlet{colorSignalQuart}{black}
 \colorlet{colorSignalBas}{cyan}

 \tikzstyle{carreUpRight}=[pattern=north west lines,pattern color=black]
 \tikzstyle{carreDownRight}=[pattern=north west lines,pattern color=black]
 \tikzstyle{carreUpLeft}=[color=colorCarreUpLeft]
 \tikzstyle{carreDownLeft}=[color=colorCarreDownLeft]
 \tikzstyle{carreBorder}=[color=colorCarreBorder]
 \tikzstyle{troisQuart}=[color=colorTroisQuart]
 \tikzstyle{unQuartBas}=[pattern=dots]
 \tikzstyle{signalCarres}=[color=colorSignalCarres]
 \tikzstyle{signalBas}=[color=colorSignalBas]
 \tikzstyle{signalQuart}=[color=colorSignalQuart]
 
 \tikzstyle{bGcarreUpRight}=[fill=colorCarreUpRight]
 \tikzstyle{bGcarreUpLeft}=[color=colorCarreUpLeft]
 \tikzstyle{bGcarreDownLeft}=[color=colorCarreDownLeft]
 \tikzstyle{bGcarreDownRight}=[color=colorCarreDownRight]
 \tikzstyle{bGcarreBorder}=[color=colorCarreBorder]
 \tikzstyle{bGtroisQuart}=[color=colorTroisQuart]
 \tikzstyle{bGunQuartBas}=[color=colorUnQuartBas]
 \tikzstyle{bGsignalCarres}=[color=colorSignalCarres]
 \tikzstyle{bGsignalBas}=[color=colorSignalBas]
 \tikzstyle{bGsignalQuart}=[color=colorSignalQuart]

\date{}

\newcommand{\ZZ}{\mathbb{Z}}

\newcommand{\NN}{\mathbb{N}}

\newcommand{\pizu}{$\Pi^0_1$\xspace}
\newcommand{\cantor}{\left\{0,1\right\}^\NN}

\newcommand{\turinf}{\leq_T}

\newcommand{\turequiv}{\equiv_T}
\newcommand{\turdeg}{\deg_T}

\definecolor{darkgreen}{RGB}{44,111,57}

\newcommand{\idest}{\textit{i.e.}}

\begin{document}
\title{Turing degrees of multidimensional SFTs}

\author{Emmanuel Jeandel\thanks{mail: Emmanuel.Jeandel@lif.univ-mrs.fr} \and 
Pascal Vanier \thanks{mail: Pascal.Vanier@lif.univ-mrs.fr}\\Laboratoire
d'informatique fondamentale de Marseille (LIF)\\
Aix-Marseille Universit\'e, CNRS\\
39 rue Joliot-Curie, 13453 Marseille Cedex 13, FRANCE}

\maketitle

\begin{abstract}
In this paper we are interested in computability aspects of subshifts and in
particular Turing degrees of 2-dimensional SFTs (i.e. tilings).
To be more precise, we prove that given any \pizu class $P$ of $\cantor$
there is a SFT $X$ such that $P\times\ZZ^2$
is recursively homeomorphic to $X\setminus U$ where $U$ is a computable set of
points.  As a consequence, if $P$ contains a computable member, $P$ and $X$ have
the exact same set of Turing degrees. On the other hand, we prove that if $X$
contains only non-computable members, some of its members always have
different but comparable degrees. This gives a fairly complete study of
Turing degrees of SFTs.
\end{abstract}
%\maketitle

%\section*{Introduction}\label{S:intro}

Wang tiles have been introduced by Wang \cite{WangII} to study fragments of 
first order logic. Independently, subshifts of finite type (SFTs) were
introduced to study dynamical systems. From a computational and dynamical
perspective,
SFTs and Wang tiles are equivalent, and most recursive-flavoured results about
SFTs were proved in a Wang tile setting. 

Knowing whether a tileset can tile the plane with a given tile at the 
origin (also known as the origin constrained domino problem) was proved
undecidable by Wang \cite{Wang3}. Knowing whether a tileset can tile the plane
in the general case was proved undecidable by Berger \cite{BergerPhd,Berger2}.

Understanding how complex, in the sense of recursion theory, the 
points of an SFT can be is a question that was first studied by Myers 
\cite{Myers} in 1974. Building on the work of Hanf \cite{Hanf1}, he gave a
tileset with no computable tilings. Durand/Levin/Shen \cite{BDLS} showed, 40
years later, how to build a tileset for which all
tilings have high Kolmogorov complexity.

A \pizu class (of sets) is an effectively closed subset of $\cantor$, or
equivalently the set of oracles on which a given Turing machine does not halt.
\pizu classes occur naturally in various areas in computer science and
recursive mathematics, see e.g. \cite{CenzerRec, SimpsonSurvey} 
and the upcoming book \cite{Cenzerbook}. 
It is easy to see that any SFT is a \pizu class
(up to a computable coding of $\Sigma^{\mathbb{Z}^2}$ into $\cantor$).
This has various consequences. As an example, every non-empty SFT
contains a point which is not Turing-hard 
(see Durand/Levin/Shen \cite{BDLS} for a
self-contained proof). The main question is how different SFTs
are from \pizu classes. In the one-dimensional case,
some answers to these questions were given by Cenzer/Dashti/King/Tosca/Wyman \cite{Dashti,
CenzerShift,CenzerTocs}.

The main result in this direction was obtained by Simpson \cite{MEdv},
building on the work of Hanf and Myers: for
every \pizu class $S$, there exists a SFT with the
same \emph{Medvedev} degree as $S$. The Medvedev degree roughly relates to the
``easiest'' Turing degree of $S$. What we are interested in 
is a stronger result: \emph{can we find for every \pizu class $S$ a SFT which has
the same Turing degrees?} We prove in this article that
this is true if $S$ contains a computable point but not always when this is not
the case. 
More exactly we build (Theorem \ref{thm:pizuhomtiling}) for every \pizu class $S$ 
a SFT for which the set of Turing degrees is exactly the same as 
for $S$ with the additional Turing degree of computable points.
We also show that SFTs that do not contain any computable point
always have points with different but comparable degrees 
(Corollary~\ref{cor:compdiffdegrees}), a property that is not true for all \pizu 
classes. In particular there exist \pizu classes that do
not have any points with comparable degrees.

As a consequence, as every \emph{countable} \pizu class contains a computable
point, the question is solved for countable sets: the sets of
Turing degrees of countable \pizu classes are the same as the sets of Turing
degrees of countable sets of tilings. In particular, there exist countable
sets of tilings with some non-computable points. This can be thought as a
two-dimensional version of Corollary~4.7 in \cite{CenzerTocs}.

This paper is organized as follows. After some preliminary definitions, we start with a
quick proof of a generalization of Hanf, already implicit in Simpson~\cite{MEdv}. 
We then build a very specific tileset, which forms a
grid-like structure while having only countably many tilings, all of them computable.
This tileset will then serve as the main ingredient to prove the result on the case
of classes with a computable point in section~\ref{S:pizuwithrec}.
In section~\ref{S:pizuwithoutrec} we finally show the result on classes without 
computable points.

\section{Preliminaries}\label{S:defs}
 \subsection{SFTs and tilings}
We give here some standard definitions and facts about multidimensional
subshifts, one may consult Lind~\cite{LindMulti} for more details.
 Let $\Sigma$ be a finite alphabet, the $d$-dimensional
full shift on $\Sigma$ is the set
$\Sigma^{\ZZ^d}=\left\{c=(c_x)_{x\in\ZZ^d} \middle\vert \forall x\in\ZZ^d,
c_x\in\Sigma\right\}$. For $v\in\ZZ^d$, the shift functions
$\sigma_v:\Sigma^{\ZZ^d}\to\Sigma^{\ZZ^d}$, are defined locally by
$\sigma_v(c_x)=c_{x+v}$. The full shift equipped with the distance
$d(x,y)=2^{-\min\left\{\left\|v\right\|\middle\vert v\in\ZZ^d,x_v\neq
y_v\right\}}$ is a compact, perfect, metric space on which the shift functions act
as homeomorphisms. An element of $\Sigma^{\ZZ^d}$ is called a
\emph{configuration}.

Every closed shift-invariant (invariant by application of any $\sigma_v$)
subset $X$ of $\Sigma^{\ZZ^d}$ is called a \emph{subshift}. An element of a
subshift is called a point of this subshift.

Alternatively, subshifts can be defined with the help of forbidden patterns. A
\emph{pattern} is a function $p:P \to \Sigma$, where $P$ is a finite subset of
$\ZZ^d$. Let $\mathcal F$ be a collection of \emph{forbidden} patterns, the
subset $X_F$ of $\Sigma^{\ZZ^d}$ containing only configurations having nowhere a
pattern of $F$. More formally, $X_{\mathcal F}$ is defined by
$$X_{\mathcal F}=\left\{x\in\Sigma^{\ZZ^d}\middle\vert \forall z\in\ZZ^d,\forall
p\in \mathcal F, x_{z+P}\neq p \right\}\textrm{.}$$

In particular, a subshift is said to be a \emph{subshift of finite type} (SFT)
when the collection of forbidden patterns is finite. Usually, the patterns used
are \emph{blocks} or \emph{$n$-blocks}, that is they are defined over a finite
subset $P$ of $\ZZ^d$ of the form ${\llbracket 0,n-1\rrbracket}^d$.

Given a subshift $X$, a block or
pattern $p$ is said to be \emph{extensible} if there exists $x\in X$ in which
$p$ appears, $p$ is also said to be extensible to $x$. 

In the rest of the paper, we will use the notation $\Sigma_X$ for the alphabet
of the subshift $X$.

 A subshift $X\subseteq\Sigma_X^{\ZZ^2}$ is a \emph{sofic shift}
if and only if there exists a SFT $Y\subseteq\Sigma_Y^{\ZZ^2}$  and a map
$f:\Sigma_Y\rightarrow \Sigma_X$ such that for any point $x\in X$, there exists
a point $y\in Y$ such that for all $z\in\ZZ^d, x_z=f(y_z)$.
 
 \emph{Wang tiles} are unit squares with colored edges which may not be flipped
or rotated. A \emph{tileset} $T$ is a finite set of Wang tiles. A
\emph{coloring of the plane} is a mapping $c:\ZZ^2\rightarrow T$ assigning a
Wang tile to each point of the plane. If all adjacent tiles of a
 coloring of the plane have matching edges, it is called a
tiling. 

 In particular, the set of tilings of a Wang tileset is a SFT on the alphabet
formed by the tiles. Conversely, any SFT is isomorphic to a Wang tileset. From a
recursivity point of view, one can say that SFTs and Wang tilesets are
equivalent. In this paper, we will be using both indiscriminately. In
particular, we denote by $X_T$ the SFT associated to a set of tiles $T$.
 
We say a SFT (tileset) is \emph{origin constrained} when the letter (tile) at
position $(0,0)$ is forced, that is to say, we only look at the valid tilings
having a given letter (tile) $t$ at  the origin.
 
 More information on SFTs may be found in Lind and  Marcus' book
\cite{LindMarkus}.

 The notion of \emph{Cantor-Bendixson derivative} is defined on set of
configurations. This notion was introduced for tilings by
Ballier/Durand/Jeandel \cite{BDJSTACS}. A configuration $c$ is said to be
\emph{isolated} in a set of configurations
 $C$ if there exists a pattern $p$ such that $c$ is the only configuration of
$C$ containing $p$.
 The Cantor-Bendixson derivative of $C$ is denoted  by $D(C)$ and consists of all
configurations of $C$
 except the isolated ones. We define $C^{(\lambda)}$ inductively for any ordinal
$\lambda$:
 \begin{itemize}
   \item $C^{(0)}=S$
   \item $C^{(\lambda+1)}=D\left( C^{\left( \lambda\right)}) \right)$
   \item $C^{(\lambda)} = \bigcap_{\gamma<\lambda} C^{(\gamma)}$ when $\lambda$
is a limit ordinal.
 \end{itemize}

 The \emph{Cantor-Bendixson rank} of $C$, denoted by $CB(C)$, is defined as the
first ordinal $\lambda$ such that
 $C^{(\lambda)}=C^{(\lambda+1)}$. If $C$ is countable, then $C^{CB(C)}$ is empty. 
An element $x$ is of rank $\lambda$ in $C$ if
$\lambda$ is the  least ordinal such that $x\not\in C^{(\lambda)}$.
 
 A configuration $x$ is \emph{periodic}, if there exists $n\in\NN^{*}$ such that
$\sigma_{ne_i}(x)=x$, for any $i\in\{1,\dots,d\}$, where the $e_i$'s form the
standard basis. A \emph{vector of periodicity} of a configuration is a vector
$v\in\ZZ^d\setminus\{(0,\dots,0)\}$ such that $\sigma_v(x)=x$.
 A configuration $x$ is \emph{quasiperiodic} (see Durand \cite{Durand1999}
for instance) if for any pattern $p$ appearing in $x$, there exists $N$ such
that this pattern appears in all $N^d$ cubes in $x$. In particular, a periodic
point is quasiperiodic. A configuration
is \emph{strictly quasiperiodic} if it is quasiperiodic and not periodic.
A subshift is \emph{minimal} if it is non-empty and contains no proper
non-empty subshift. Equivalently, all its points have the same
patterns. In this case, it contains only quasiperiodic points.
 It is known that every subshift contains a minimal subshift, see
  e.g. Durand \cite{Durand1999}.

\subsection{Computability background}
 A \emph{\pizu class} $P\subseteq\cantor$ is a class of infinite sequences
on $\{0,1\}$ for which there exists a Turing machine that given $x\in\left\{
0,1 \right\}^\NN$ as an oracle halts if and only if $x\not\in P$. Equivalently,
a class $S\subseteq\cantor$ is \pizu if there exists a computable set $L$
so that $x\in S$ if and only if no prefix of $x$ is in $L$. An element of a \pizu class is
called a \emph{member} of this class.

 We say that two sets $S,S'$ are \emph{recursively homeomorphic} if there exists
a bijective computable function $f:S\rightarrow S'$. That is to say there are
two Turing machines $M$ (resp. $M'$) such that given a member of $S$ (resp.
$S'$) computes a member of $S'$ (resp. $S$). Furthermore, for any $s\in S,s'\in
S'$ such that $s'$ is computed by $M$ from $s$, $M'$ computes $s$ from $s'$.

The \emph{Cantor-Bendixson rank} of $S$, is well defined similarly as for
subshifts.
 
 See Cenzer/Remmel~\cite{CenzerRec} for \pizu classes and Kechris~\cite{Kechris} 
 for Cantor-Bendixson rank and derivative. 
 
 For $x,y\in\cantor$ we say that $x$ is \emph{Turing-reducible} to $y$ 
 if $y$ is computable by a Turing machine using $x$ as an oracle and we write
$y\turinf x$. If $x\turinf y$ and $y\turinf x$, we say that $x$ and $y$ are 
\emph{Turing-equivalent} and we write $x\turequiv y$. The \emph{Turing degree}
of $x\in\cantor$, denoted by $\turdeg x$, is its equivalence class under the relation $\turequiv$.

\subsection{Subshifts and \pizu classes}
 
As is clear from the definitions, SFTs in any dimension are \pizu
classes. More generally, \emph{effective} subshifts, see e.g. Cenzer/Dashti/King
\cite{CenzerShift}), that is subshifts defined by a computable (or
equivalently, in this case, by a computably enumerable) set of forbidden
patterns are \pizu classes. As such, they enjoy similar properties.
In particular, there exist many ``basis theorems'', \idest theorems that assert
that
any \pizu (non-empty) class has a member with some specific property.

As an example, every countable \pizu class has a computable member,
see e.g. Cenzer/Remmel \cite{Cenzerbook}.
For subshifts, we can say a bit more: every countable subshift has a
periodic (hence computable) member.

Some basis theorems for \pizu classes can be easily reproven in the context
of subshifts: The proof that every \pizu class has a point of
\emph{low} degree (the formal definition is not important here, but it
can be interpreted as ``nearly computable'') \cite{JockuschSoare72a} was reproven for  subshifts
(actually tilings) in Durand/Levin/Shen \cite{BDLS}.

 \section{\pizu classes and origin constrained tilings}\label{S:hanf}
A straighforward corollary  of Hanf \cite{Hanf1} is that every \pizu class is recursively
homeomorphic to an origin constrained SFTs and conversely. This is stated
explicitly in Simpson \cite{MEdv}.

\begin{figure}
 \begin{center}
\scalebox{0.8}{
\begin{tikzpicture}[auto,scale=0.5]
\filldraw[fill=blue!20] (0,0) rectangle (4,4);
\node[draw,circle] (etat) at (1,2) {$s$};
\draw[-latex] (1,0) node[below] (basetat) {$s$} -- (etat);
\draw[-latex] (etat) -- (1,3) -- (0,3);
%\draw (1,3) -- (etat);
\node[draw,rectangle,fill=white] (ruban) at (2,2) {$a$};
\draw[-latex] (2,1) -- (ruban);
\node[left] at (0,3) (hautetat) {$s'$};
\node[above] at (2,4) (hautruban) {$a'$};
\node[below] at (2,0) (basruban) {$a$};
%\node[left] at (0,2) (gaucheetat) {$s$};
\begin{scope}[shift={(6,0)}]
\filldraw[fill=blue!20] (0,0) rectangle (4,4);
\node[draw,circle] (etat) at (1,2) {$s$};
\draw[-latex] (1,0) node[below] (basetat) {$s$} -- (etat);
\draw[-latex] (etat) -- (1,4);
%\draw (1,3) -- (etat);
\node[draw,rectangle,fill=white] (ruban) at (2,2) {$a$};
\draw[-latex] (2,1) -- (ruban);
\node[above] at (1,4) (hautetat) {$s'$};
\node[above] at (2,4) (hautruban) {$a'$};
\node[below] at (2,0) (basruban) {$a$};
%\node[left] at (0,2) (gaucheetat) {$s$};
\end{scope}
\begin{scope}[shift={(12,0)}]
\filldraw[fill=blue!20] (0,0) rectangle (4,4);
\node[draw,circle] (etat) at (1,2) {$s$};
\draw[-latex] (1,0) node[below] (basetat) {$s$} -- (etat);
\draw[-latex] (etat) -- (1,3) -- (4,3);
%\draw (1,3) -- (etat);
\node[draw,rectangle,fill=white] (ruban) at (2,2) {$a$};
\draw[-latex] (2,1) -- (ruban);
\node[right] at (4,3) (hautetat) {$s'$};
\node[above] at (2,4) (hautruban) {$a'$};
\node[below] at (2,0) (basruban) {$a$};
%\node[left] at (0,2) (gaucheetat) {$s$};
\end{scope}
\begin{scope}[shift={(18,0)}]
\filldraw[fill=blue!20] (0,0) rectangle (4,4);
\draw[-latex]  (0,3) -- (1,3) -- (1,4);
%\draw (1,3) -- (etat);
\node[draw,rectangle,fill=white] (ruban) at (2,2) {$a$};
\node[left] at (0,3)  (basetat) {$s$};
\node[above] at (1,4) (hautetat) {$s$};
\node[above] at (2,4) (hautruban) {$a$};
\node[below] at (2,0) (basruban) {$a$};
\end{scope}
\begin{scope}[shift={(0,-6)}]
\filldraw[fill=blue!20] (0,0) rectangle (4,4);
\draw[-latex]  (4,3) -- (1,3) -- (1,4);
%\draw (1,3) -- (etat);
\node[draw,rectangle,fill=white] (ruban) at (2,2) {$a$};
\node[right] at (4,3)  (basetat) {$s$};
\node[above] at (1,4) (hautetat) {$s$};
\node[above] at (2,4) (hautruban) {$a$};
\node[below] at (2,0) (basruban) {$a$};
\end{scope}
\begin{scope}[shift={(6,-6)}]
\filldraw[fill=blue!20] (0,0) rectangle (4,4);
%\draw (1,3) -- (etat);
\node[draw,rectangle,fill=white] (ruban) at (2,2) {$a$};
\node[above] at (2,4) (hautruban) {$a$};
\node[below] at (2,0) (basruban) {$a$};
\end{scope}
\begin{scope}[shift={(12,-6)}]
\fill[color=blue!20] (0,1) rectangle (4,4);
\draw (0,0) rectangle (4,4);
\node[draw,circle] (etat) at (1,2) {$s_0$};
\draw[-latex] (etat) -- (1,3) -- (4,3);
\draw[-latex] (0,2) -- (etat);
\node[draw,rectangle,fill=white] (ruban) at (2,2) {$a$};
\draw[-latex] (2,1) -- (ruban);
\node[left] at (0,2) (basetat) {$s_0$};
\node[right] at (4,3) (hautetat) {$s'$};
\node[above] at (2,4) (hautruban) {$a'$};
\end{scope}
\begin{scope}[shift={(18,-6)}]
\fill[color=blue!20] (0,1) rectangle (4,4);
\draw (0,0) rectangle (4,4);
\node[draw,circle] (etat) at (1,2) {$s_0$};
\draw[-latex] (etat) -- (1,4);
\draw[-latex] (0,2) -- (etat);
\node[draw,rectangle,fill=white] (ruban) at (2,2) {$a$};
\draw[-latex] (2,1) -- (ruban);
\node[above] at (1,4) (hautetat) {$s'$};
\node[above] at (2,4) (hautruban) {$a'$};
\node[left] at (0,2) (basetat) {$s_0$};
\end{scope}
\begin{scope}[shift={(0,-12)}]
\fill[color=blue!20] (0,1) rectangle (4,4);
\draw (0,0) rectangle (4,4);
%\draw (1,3) -- (etat);
\draw[-latex] (0,3) -- (1,3) -- (1,4);
\node[draw,rectangle,fill=white] (ruban) at (2,2) {$B$};
\node[above] at (1,4) (hautetat) {$s$};
\node[left] at (0,3) (basetat) {$s$};
\node[above] at (2,4) (hautruban) {$B$};
\end{scope}
\begin{scope}[shift={(6,-12)}]
\fill[color=blue!20] (0,1) rectangle (4,4);
\draw (0,0) rectangle (4,4);
%\draw (1,3) -- (etat);
%\draw (4,3) -- (1,3) -- (1,4);
\node[draw,rectangle,fill=white] (ruban) at (2,2) {$B$};
%\node[above] at (1,4) (hautetat) {$s$};
\node[above] at (2,4) (hautruban) {$B$};
\end{scope}
\begin{scope}[shift={(12,-12)}]
\fill[color=blue!20] (4,0) rectangle (1,4);
\draw (0,0) rectangle (4,4);
\end{scope}
\begin{scope}[shift={(18,-12)}]
\fill[color=blue!20] (4,1) rectangle (1,4);
\draw (0,0) rectangle (4,4);
\draw[-latex] (2,2) -- (4,2);
\node[right] at (4,2) (basetat) {$s_0$};
\end{scope}
\end{tikzpicture}
}
 
 \end{center}
 \caption{A set of Wang tiles, encoding computation of a Turing
machine: the states are in the circles and the tape is in the rectangles. The bottom right tile
starts computations. A
tiling containing this tile contains the space-time diagram of some
run of the Turing machine.}
 \label{fig:tmencoding}
\end{figure}
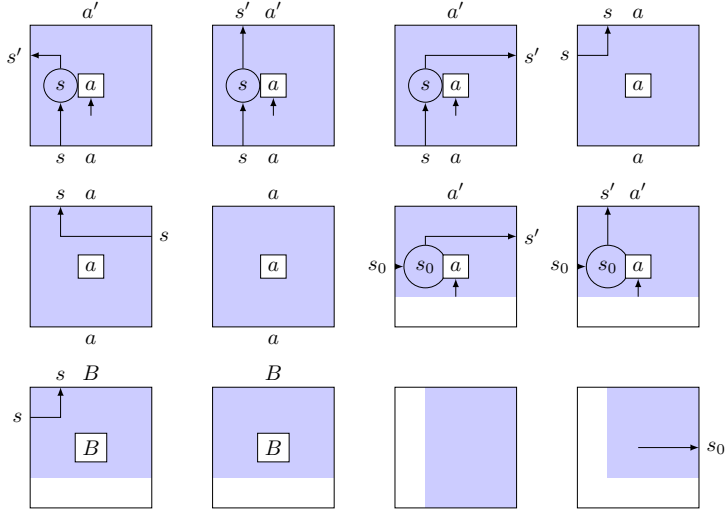

\begin{theorem}
  Given any \pizu class $P\subseteq\cantor$, there exists a SFT $X$ and a
letter $t\in\Sigma_X$ such that each origin constrained point corresponds to a
member of $P$.
\end{theorem}
\begin{proof}
Let $P$ be a \pizu class, and $M$ the Turing machine that proves it,
that is $M$ given $x\in \cantor$ as an oracle halts if and only if
$x\not\in P$.

    We use the classic encoding of Turing machines as Wang tiles, see
fig.~\ref{fig:tmencoding}.  We modify all tiles containing a symbol
from the tape, to allow them to contain a second symbol. This symbol
is copied vertically. All these second symbols represent the oracle. 

Then the SFT constrained by the tile starting the computation contains
exactly the runs of the Turing machine with members of $P$ on the
oracle tape.
\end{proof}

\begin{corollary}
  Any \pizu class $P$ of $\cantor$ is recursively homeomorphic to an origin
constrained
  SFT.
\end{corollary}
 
\section{Producing a sparse grid}\label{S:tileset}

The main problem in the previous construction is that points which do not
have the given letter at the origin can be very wild: they may correspond to
configurations with no computation (no head of the Turing Machine) or
computations starting from an arbitrary (not initial) configuration. A way to
solve this problem is described in Myers' paper \cite{Myers} but is unsuitable
for our purposes (It was however used by Simpson
to obtain a weaker result on Medvedev degrees, see   \cite{MEdv}).
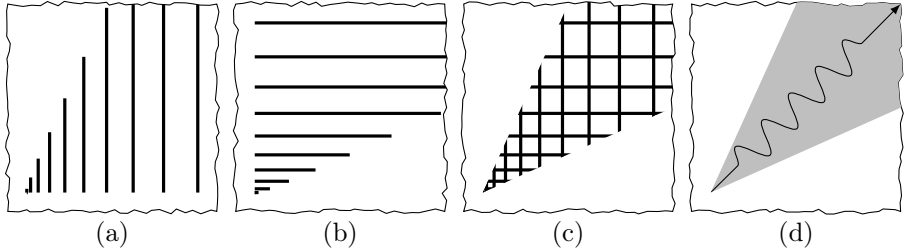
\begin{figure}[htbp]
 \begin{center}
 \begin{tikzpicture}[scale=0.025]
  \def\lx{10};
  \def\linewidth{1.2pt};
  \begin{scope}
  \begin{scope}
    \clip[draw,decorate,decoration={random steps, segment length=3pt,
amplitude=1pt}] (-10,-10) rectangle (100,100);
    \foreach \i in {0,...,\lx}{
      \def\x{\i*\i-\i};
      \draw[line width=\linewidth] (\x,0) -- (\x,\i*\i*2);
    }
  \end{scope}
  \node[below] at (45,-10) {(a)};
  \end{scope}
  \begin{scope}[shift={(120,0)}]
  \begin{scope}
    \clip[draw,decorate,decoration={random steps, segment length=3pt,
amplitude=1pt}] (-10,-10) rectangle (100,100);
    \foreach \i in {0,...,\lx}{
      \def\x{\i*\i-\i};
      \draw[line width=\linewidth] (0,\x) -- (\i*\i*2,\x);
    }
  \end{scope}
  \node[below] at (45,-10) {(b)};
  \end{scope}
  \begin{scope}[shift={(240,0)}]
  \begin{scope}
    \clip[draw,decorate,decoration={random steps, segment length=3pt,
amplitude=1pt}] (-10,-10) rectangle (100,100);
    \begin{scope}
      \clip (0,0) -- (\lx*\lx*2,\lx*\lx-\lx) -- (\lx*\lx-\lx,\lx*\lx*2) --
cycle;
      \foreach \i in {0,...,\lx}{
        \def\x{\i*\i-\i};
        \draw[line width=\linewidth] (0,\x) -- (\i*\i*2,\x);
        \draw[line width=\linewidth] (\x,0) -- (\x,\i*\i*2);
      }
    \end{scope}
  \end{scope}
  \node[below] at (45,-10) {(c)};
  \end{scope}
   \begin{scope}[shift={(360,0)}]
  \begin{scope}[even odd rule]
    \clip[draw,decorate,decoration={random steps, segment length=3pt,
amplitude=1pt}] (-10,-10) rectangle (100,100);
    \begin{scope}
      \clip (0,0) -- (\lx*\lx*2,\lx*\lx-\lx) -- (\lx*\lx-\lx,\lx*\lx*2) --
cycle;
      \fill[color=lightgray] (-10,-10) rectangle (110,110);
%      \foreach \i in {0,...,\lx}{
%        \def\x{\i*\i-\i};
%        \draw[line width=\linewidth] (\i*\i*2,\x)  -- (\x,\i*\i*2);
%      }
      \draw[-latex,decorate,decoration={snake,amplitude=2mm,segment
length=5mm,pre length=5mm,post length=5mm}]
      (0,0) -- (\lx*\lx,\lx*\lx);

    \end{scope}
  \end{scope}
  \node[below] at (45,-10) {(d)};
  \end{scope}

 \end{tikzpicture}
 \end{center}
 \caption{The tiling in which the Turing machines will be encoded.}
 \label{fig:layers}
\end{figure}

Our idea is as follows: we build a SFT which will contain, among other points,
the \emph{sparse grid} of Figure~\ref{fig:layers}c. 
The interest being that all other points will have at most
one intersection of two black lines.
This means that if we put computation cells of a given
Turing machine in the intersection points, 
every point which is not of the form of Figure~\ref{fig:layers}c
will contain at most one cell of the Turing machine, and thus will contain no
computation.

\begin{figure}[htbp]
\begin{center}
\includegraphics{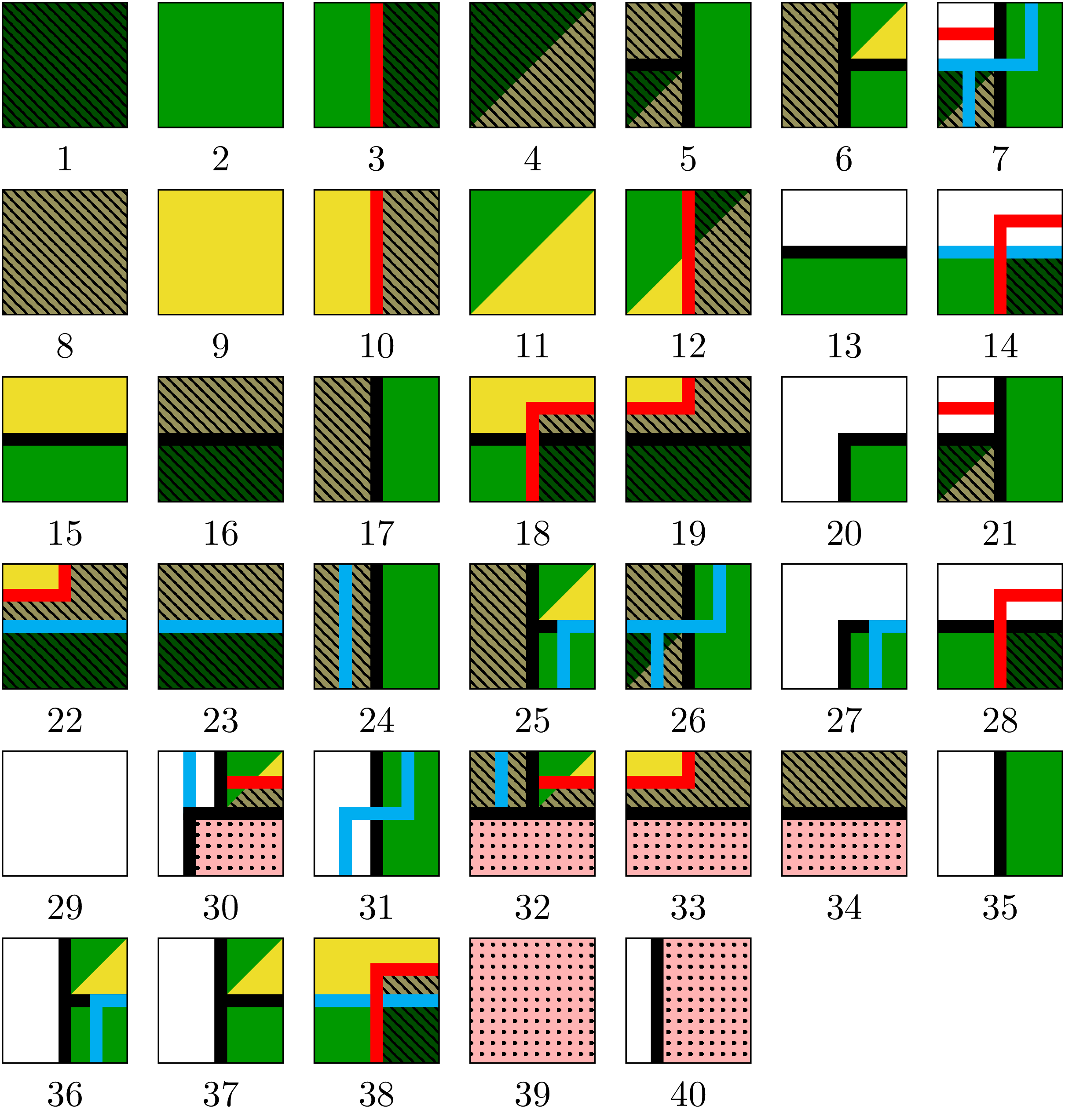}
\end{center}
\caption{Our set of Wang tiles $T$.}
\label{fig:tileset}
\end{figure}

To do this construction, we will first draw increasingly big
and distant columns as in Figure~\ref{fig:layers}a  and then
superimpose the same construction for rows as in Figure~\ref{fig:layers}b,
thus obtaining the grid of Figure~\ref{fig:layers}c.

It is then fairly straightforward to see how we can encode a Turing machine
inside a configuration
having the skeleton of Figure~\ref{fig:layers}c by looking at it diagonally:
time increases going
to the north-east and the tape is written on the north-west/south-east 
diagonals\footnote{Note that we have to wait for the diagonal to increase to have a new step of
computation, in order to have enough space on the tape.}.

Our set of tiles $T$ of Figure~\ref{fig:tileset} gives the skeleton of
Figure~\ref{fig:layers}a 
when forgetting everything but the black vertical borders. We will prove in this
section that it
is countable.
We set here the vocabulary:
\begin{itemize}
  \item tile 30 is the \emph{corner tile}
  \item tile 20 and 27 are the \emph{top tiles}
  \item tiles 30, 32, 33, 34 are the \emph{bottom tiles}
  \item a \emph{vertical line} is formed of a vertical succession of tiles containing a
vertical black line (tiles 5, 6, 7, 17, 21, 24, 25, 26, 31, 35, 36, 37), which may be ended
by bottom and/or top tiles. 
  \item a \emph{horizontal line} is formed of a horizontal succession of tiles
containing a horizontal black line (tiles 13, 14, 15, 16, 18, 19, 22, 23, 28, 32, 33, 34, 38), and
may be ended by tiles 5,6,7,25,26,36,37, thus forcing a vertical line at this end,
  \item a \emph{diagonal} is a diagonal succession (positions (i,j),(i+1,j+1),...) of
tiles among 4,12,11,
  \item a \emph{square} is a $\llbracket 0,k\rrbracket^2$ valid tiling such that $\{0\}\times
\llbracket 1,k-1\rrbracket$ and  $ \{k\}\times \llbracket 1,k-1\rrbracket$ are vertical lines, and
$\llbracket 1,k-1\rrbracket\times\{0\}$ and $\llbracket 1,k-1\rrbracket\times\{k\}$ are horizontal
lines. Remark that the color on the right of the first column and on the left of the last one force
the existence of a counting signal inbetween and of a diagonal tile on each of the $(i,i)$ positions
for
$0<i<k$.
  
  \item the \emph{increase signal}
  \tikz[scale=0.25]{\fill[bGsignalBas] (0,0) rectangle (1,0.2);\fill[signalBas]
(0,0) rectangle (1,0.2);}
  is formed by a path of tiles among 7, 14, 22, 23, 24, 25, 26, 27, 30, 31, 36, 38, such that the
blue signal is connected, this signal will force the squares to increase in size by exactly one in
each column.
  \item the \emph{counting signal}
  \tikz[scale=0.25]{\fill[bGsignalCarres] (0,0) rectangle
(0.2,1);\fill[signalCarres] (0,0) rectangle (0.2,1);}
  is formed by a path of tiles among 3 ,7, 10, 12, 14, 19, 22, 32, 33, 38, such that the counting
signal
is connected. It may be ended only by tiles 30,32 and 7,21. This signal will force the number of
squares in each column to be at most the size of these squares.
\end{itemize}

Note that whenever the corner tile appears in a
point, it is necessarily a shifted copy of the point on
Figure~\ref{fig:confalpha}: the corner tile forces the tiles on its right to be bottom tiles and
the first above to be tile 33, then on top of it must be tile 31 and then tile 27. This forces the
existence of the first square, \idest the first column. Then the increase signal forces the second
column to start with a square of size increased by one, and thus to have exactly one more square
(increase signal), and so on.

\begin{lemma}\label{lem:twolines}
  The SFT $X_T$ admits at most one point, up to translation, with two or more
vertical lines. This point is drawn on Figure~\ref{fig:confalpha}.
\end{lemma}
\begin{proof}

The idea of the construction is to force that whenever there are two
vertical lines, then the point is a shifted copy of the one in
Figure~\ref{fig:confalpha}.

Suppose that we have a tiling in which two vertical lines appear. They may be ended on their bottom
only by a bottom tile 30 or 32, but when a bottom tile appears, it forces all tiles to its right to
be bottom tiles. Because the color on each side of the vertical lines is not the same they
necessarily are connected by horizontal lines, which must form squares due to the diagonal. Suppose
the two vertical
lines are at distance $k$, then there are exactly $k$ squares between them vertically, because
of the counting signal: it must appear in each square, and is shifted to the right
every time
it crosses a horizontal line, it may be ended in each column only by tiles 32 (or 30 if it is the
leftmost column) in the bottommost square and
by 21 (resp. 7) in the topmost.

The bottommost square must have an increase signal as its top
horizontal line, since the lower left corner 32 (or 30 in case it is the leftmost square) forces the
left side to be formed of a succession of tiles 24 ended by tile 26 (resp. only one tile 31), then
the top left corner is necessarily a tile 25 (resp. 27). This forces the size of the
squares on its left to be $k-1$ and on its right to be $k+1$.

If we focus only on  the bottommost squares, they are of decreasing size when going left, the
last one is of size 1, and necessarily has the corner tile as its lower left corner.
\end{proof}

\begin{figure} 
 \begin{center}
 \includegraphics{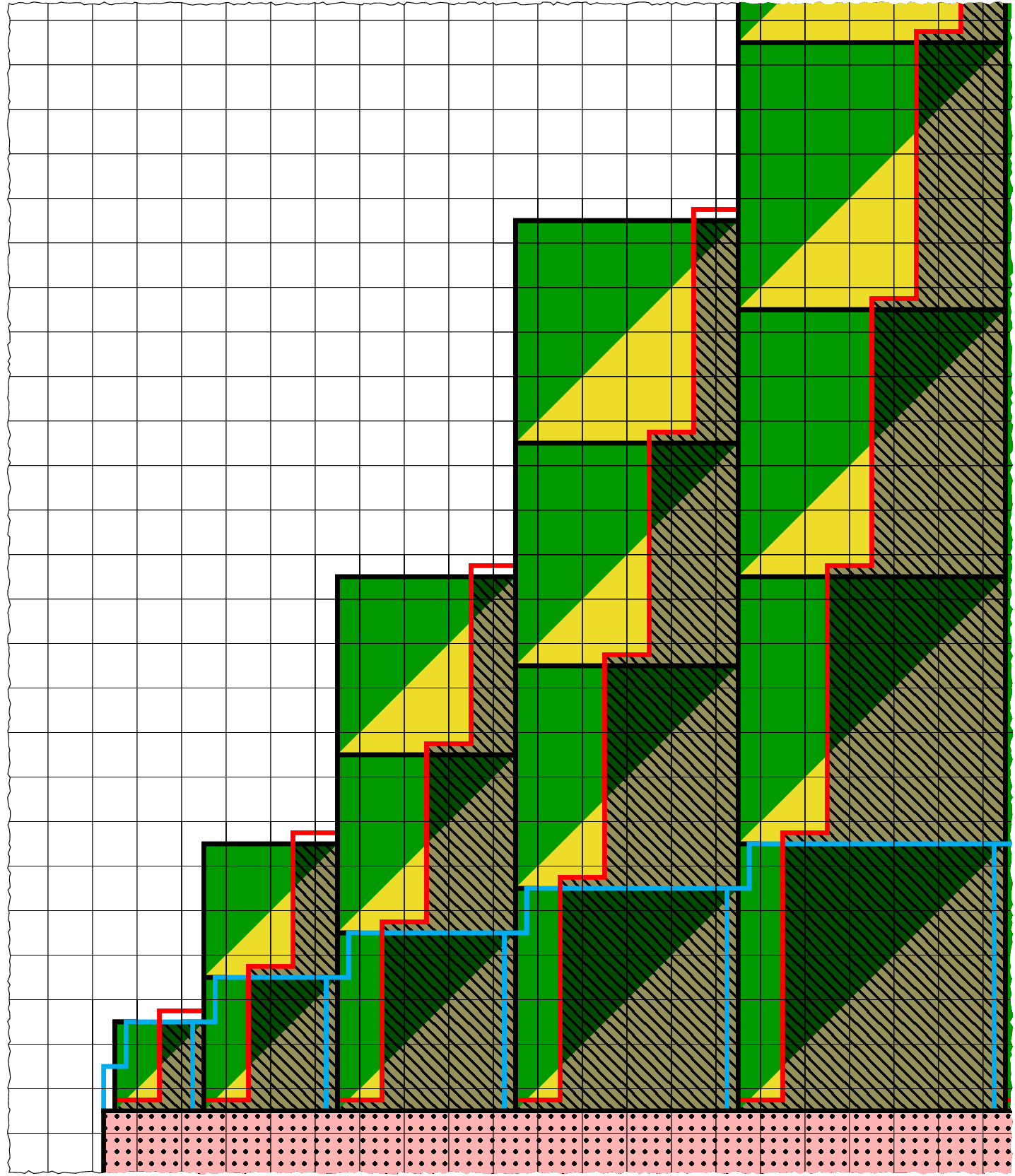}
 \end{center}
 \caption{Tiling $\alpha$: the unique valid tiling of $T$ in which there are 2
or more vertical
 lines. This tiling has Cantor-Bendixson rank 1.}
 \label{fig:confalpha}
\end{figure}

\begin{lemma}\label{lem:countable}
$X_T$ is countable.
\end{lemma}
\begin{proof}
  Lemma~\ref{lem:twolines} states that there is one point, up to shift, that has
two or more vertical lines. This means that the other points have at most one
such line.
  \begin{itemize}
    \item If a point has exactly one vertical line, then it can have at most
      two horizontal lines: one on the left of the vertical one and one on the
right. Otherwise a square would appear and the configuration would be $\alpha$.
      A counting signal may then appear on the left or the right of the vertical line
arbitrarily far from it. There is a countable number of such points.

    \item If a point has no vertical line, then it has at most one horizontal
line. A counting signal can then appear only once. There is a finite number of
such points, up to shift. 
  \end{itemize}
  There is a countable number of points that can be obtained with the tileset
$T$.
  All types of obtainable points are shown in Figures~\ref{fig:confalpha}
  and~\ref{fig:otherconf}.
  
\end{proof}

\begin{figure}
  \begin{center}
  \includegraphics{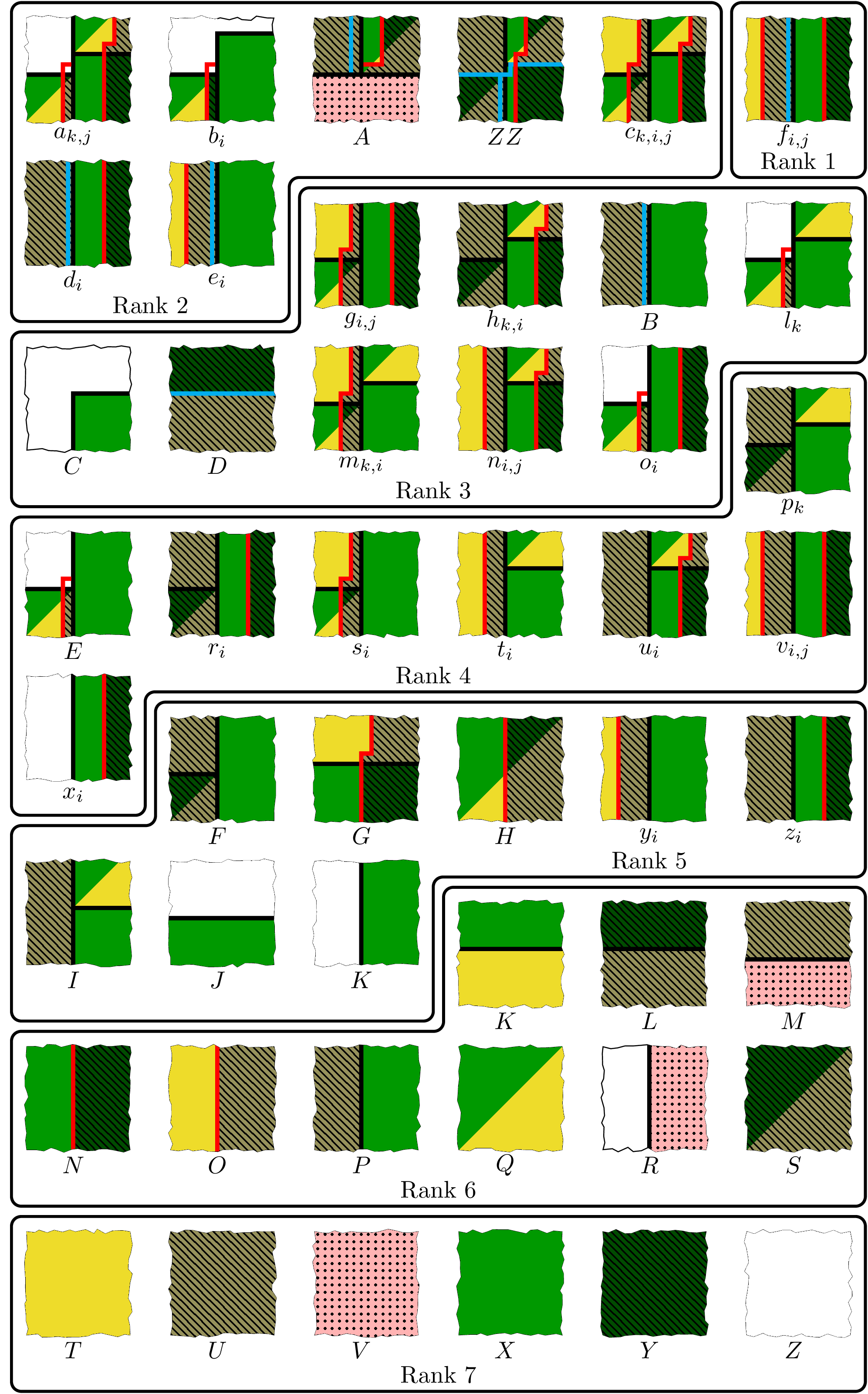}
  \end{center}
  \caption{The configurations of $X_T\setminus\{\alpha\}$: the $A-ZZ$ configurations are unique (up
to
shift), and the configurations with subscripts $i,j\in\NN,k\in\ZZ^2$ represent
the fact that distances between some of the lines (red, horizontal, vertical)
  can vary. The configurations are classified according to their Cantor-Bendixson rank. Note that
configuration $ZZ$ cannot have a counting signal on its left,
because it would force another vertical line.}
  \label{fig:otherconf}
\end{figure}

By taking our tileset $T=\{1,\dots,40\}$ and mirroring all the tiles along the
south-west/north-east diagonal,
we obtain a tileset $T'=\{1',\dots,40'\}$ with the exact same properties, except
it enforces the skeleton of
Figure~\ref{fig:layers}b. Remember that whenever the corner tile appeared in a
point, then necessarily
this point was a shifted of $\alpha$. Analogously, the corner tile of $T'$
appearing in a point means that this point is a shifted of $\alpha'$.
We hence construct a third tileset $\tau =
\left(T\setminus\left\{30\right\}\times
T'\setminus\left\{30'\right\}\right)\cup\left\{(30,30')\right\}$ which is the
superimposition of $T$ and $T'$ with the restriction that tiles $30$ and $30'$
are necessarily superimposed to each other. The corner
tile $(30,30')$ of $\tau$ has the property that whenever it appears, the tiling
is the superimposition of the skeletons of Figures~\ref{fig:layers}a
and~\ref{fig:layers}b with the corner tiles at the same
place: there is only one such tiling up to shift, we call it $\beta$. 

The skeleton of Figure~\ref{fig:layers}c is obtained from $\beta$ if we forget
about the
parts of the lines of the $T$ layer (resp. $T'$) that are superimposed to white
tiles, 29' (resp. 29), of $T'$ (resp. $T$).

As a consequence of Lemma~\ref{lem:countable}, $X_\tau$ is also countable. And
as a consequence of Lemma~\ref{lem:twolines}, the only points in $x_\tau$ in
which computation can be embedded are the shifts of $\beta$. The shape of
$\beta$ is the one of Figure~\ref{fig:layers}c, the coordinates of the points of
the grid are the following (supposing tile $(30,30')$ is at the center of the
grid):

 $$\left\{(f(n),f(m))\mid f(m)/4\leq f(n)\leq 4f(m)\right\}$$
 $$\left\{(f(n),f(m))\mid m/2\leq n \leq 2m\right\}$$

where $f(n)=(n+1)(n+2)/2-1$.

\begin{lemma}\label{lem:rankT}
  The Cantor-Bendixson rank of $X_\tau$ is less than or equal to 13.
\end{lemma}
\begin{proof}
It is clear that $(X_\tau)' \subset (X_T - \{ \alpha\} \times X_{T'} - \{
  \alpha'\})$. $X_T - \{\alpha\}$ is of rank $6$ as depicted in Figure~\ref{fig:otherconf}. 
The rank of a product being the sum of the rank (when it is finite), 
$(X_\tau)'$ is at most of rank $12$.
\end{proof}

\section{\pizu classes with computable members and SFTs}\label{S:pizuwithrec}

The SFT constructed before will allow us to prove a series of theorems on \pizu
classes with computable points. The foundation of these is
Theorem~\ref{thm:pizuhomtiling} which establishes a recursive homeomorphism
between SFTs and \pizu classes, up to a computable subset of the SFT. This
recursive homeomorphism is the best we can hope for, as will be shown in
section~\ref{S:pizuwithoutrec}. Then from this ``partial'' homeomorphism, we
will be able to deduce results on the set of Turing degrees of SFTs and \pizu
classes.

\begin{theorem}\label{thm:pizuhomtiling}
  For any \pizu class $S$ of $\cantor$ there exists a tileset
$\tau_S$ such
  that $S\times\ZZ^2$ is recursively homeomorphic to
$X_{\tau_S}\setminus O$ where $O$
  is a computable set of configurations.
\end{theorem} 
\begin{proof}
  This proof uses the construction of section~\ref{S:tileset}.
  Let $M$ be a Turing machine such that $M$ halts with $x$ as an oracle iff
$x\not\in S$.
  Take the tileset $\tau$ of section~\ref{S:tileset} and encode, as explained
earlier, in configuration $\beta$ the Turing
machine $M$ having as an oracle $x$ on an unmodifiable second tape.
  This defines a new tiling system $\tau_M$, and we define $O$ as the set of all points
 which were not constructed from a shift of $\beta$. To
each $(x,z)\in S\times\ZZ^2$ we associate the $\beta$ tiling having a corner
at position $z$
and having $x$ on its oracle tape. $O$ is computable, because it contains a
countable number (Lemma~\ref{lem:countable}) of computable points (none of
these points can contain more than one step of computation). 
\end{proof}

\begin{corollary}\label{corr:turingdegree:withrec}
 For any \pizu class $S$ of $\cantor$ with a computable member, there exists a
SFT $X$ with the same set of Turing degrees.
\end{corollary}

\begin{corollary}\label{corr:turingdegree:countable}
 For any countable \pizu class $S$ of $\cantor$, there exists
a SFT $X$ with the same set of Turing degrees.
\end{corollary}
\begin{proof}
 We know, from Cenzer/Remmel~\cite{CenzerRec}, that countable \pizu
 classes have
\textbf{0} (computable elements) in their set of Turing degrees, thus the
SFT $X_{\tau_M}$ described in the proof of Theorem~\ref{thm:pizuhomtiling} has
exactly the same set of Turing degrees as $S$.
\end{proof}

\begin{theorem}\label{thm:tilerankpizu}

For any countable \pizu class $S$ of $\cantor$ there exists
a SFT $X$ with the same set of Turing degrees such that $CB(X) =
CB(S)+c$ for some constant $c \leq 13$.
\end{theorem} 
This theorem holds when $CB(S)$ is any ordinal, finite or infinite.
\begin{proof}
 Lemma~\ref{lem:rankT} states that $X_{\tau}$ is of Cantor-Bendixson
rank $c \leq 13$
In the tileset $\tau_M$ of the previous proof, the
Cantor-Bendixson rank of the contents of the tape is exactly $CB(S)$, hence 
$CB(X_{\tau_S})=CB(S)+c$.
\end{proof}

From Ballier/Durand/Jeandel~\cite{BDJSTACS} we know that for any
subshift $X$, if $CB(X)\leq 2$, then $X$ has only computable points.
Thus an optimal construction would have to augment the Cantor-Bendixson rank by
at least 2.

\begin{corollary}\label{cor:ranksofic}
  For any countable \pizu class $S$ of $\cantor$ there exists
a sofic subshift $X$ with the same set of Turing degrees such that $CB(X) = CB(S)+2$.
\end{corollary}
\begin{proof}
 Take a projection that just keeps the symbols of the Turing machine tape in the SFT
$\tau_M$  from the proof of Theorem~\ref{thm:pizuhomtiling} and maps everything
else to a blank symbol. Recall the Turing machine tape cells are the
intersections of the vertical lines and horizontal lines. This projection leads
to 3 possible configurations (up to shift):
 \begin{itemize}
  \item A configuration with a white background and points corresponding to the
intersections in the sparse grid of Figure~\ref{fig:layers}c. This is
an isolated point, of rank $1$
  \item A completely blank configuration with only one symbol
	somewhere. This configuration is isolated once we remove the
	previous one(s), hence of rank $2$.
  \item A completely blank configuration, of rank $3$.
 \end{itemize}
\end{proof}

Note that a similar theorem in dimension one for effective rather
  than sofic subshifts is proved in
  Cenzer/Dashti/Toska/Wyman~\cite[Theorem 4.6]{CenzerTocs}.

\section{\pizu classes without computable members and
subshifts}\label{S:pizuwithoutrec}

In this section we prove that two-dimensional SFTs containing only
non-computable points have the property that they always have points with
different but comparable degrees, this is Corollary~\ref{cor:2dcompdiffdegrees}.
But we first prove this result for one-dimensional subshifts, not necessarily
of finite type, in Theorem~\ref{thm:1dcompdiffdegrees}, the proof
for two-dimensional SFTs needing only a bit more work. 

One interest of these proofs, lies in
the following theorem, proved by Jockusch and Soare:

\begin{theorem}[Jockusch,
Soare]\label{thm:JockuschSoare}
 There exist \pizu classes containing no computable member, such that any two
different members are Turing-incomparable.
\end{theorem}

The proof of this result can be found in Cenzer and Remmel's upcoming book
\cite{CenzerRec} or in the original articles by Jockusch and Soare
\cite{JockuschSoare72a,JockuschSoare72b}.

This means that one cannot expect a full recursive homeomorphism, i.e. without
removal of the computable points. Furthermore, this shows that in general, when
a \pizu class $P$ has no computable member, it is not true that one can find a
SFT with the same set of Turing degrees.

The main idea of the proof is that any subshift contains a minimal
subshift. If the subshift has no computable points (actually, no
periodic points), this minimal subshift contains only strictly
quasiperiodic points. We will then use some combinatorial properties
of this minimal subshift to obtain our results.

\subsection{One-dimensional subshifts}

We start with a technical lemma that will allow us to prove the theorem:

\begin{lemma}\label{lem:1dtwowords}
 Let $x$ be a strictly quasiperiodic point of a minimal one-dimensional subshift
$A$ and $\prec$ be an order on $\Sigma_A$. For any word $w$
extensible to $x$, there exist two words $w_0$ and $w_1$ such that:
\begin{itemize}
 \item $w$ appears exactly twice in both $w_0$ and $w_1$,
 \item let $a$ and $b$ (resp. $c$ and $d$) be the first differing letters in
the blocks directly following the first and second occurrence of $w$ in $w_0$
(resp. $w_1$), then $a\prec b$ (resp. $d\prec c$). 
\end{itemize}
See Figure~\ref{fig:1dtwowords} for an illustration of $w_0$ and $w_1$.
\end{lemma}
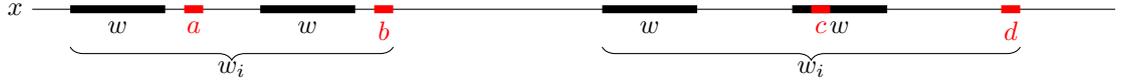
\begin{figure}
\centering
 \begin{tikzpicture}[scale=0.25]
  \draw (-7,0) node[left] {$x$} -- (50,0);
\begin{scope}[shift={(-20,0)}]
  \draw[line width=3pt] (15,0) -- +(5,0) node[midway,below] {$w$};
  \draw[line width=3pt] (25,0) -- +(5,0) node[midway,below] {$w$};
  \draw[line width=3pt,color=red] (21,0) -- +(1,0) node[midway,below] {$a$};
  \draw[line width=3pt,color=red] (31,0) -- +(1,0) node[midway,below] {$b$};
  \draw[decorate,decoration={brace,amplitude=5pt}] (32,-2) -- (15,-2)
  node[midway,below=2pt] {$w_i$};
\end{scope}
\begin{scope}[shift={(8,0)}]
  \draw[line width=3pt] (15,0) -- +(5,0) node[midway,below] {$w$};
  \draw[line width=3pt] (25,0) -- +(5,0) node[midway,below] {$w$};
  \draw[line width=3pt,color=red] (26,0) -- +(1,0) node[midway,below] {$c$};
  \draw[line width=3pt,color=red] (36,0) -- +(1,0) node[midway,below] {$d$};
  \draw[decorate,decoration={brace,amplitude=5pt}] (37,-2) -- (15,-2)
  node[midway,below=2pt] {$w_i$};
\end{scope}
  
 \end{tikzpicture}
 \caption{\label{fig:1dtwowords} Two nearest $w$ blocks, the first differing
letter $a,b$ and $c,d$ in their following blocks, and how they form the $w_i$.
Note that the first differing letter might in some cases be
inside the second occurrence of $w$, as illustrated on the right with $c,d$.
 }
\end{figure}
\begin{proof}
By quasiperiodicity of $x$, $w$ appears  infinitely many times in $x$.
By non periodicity, any two occurrences of $w$ must be followed by
eventually distinct words.
Let $y$ be the largest word so that whenever $w$ appears in $x$, it is immediately followed by
$y$. Note that $w$ appears only once in $wy$, otherwise
$x$ would be periodic.

By definition of $y$, the letters after each occurrence of $wy$ cannot
be all the same. So there exist two consecutive occurrences of
$wy$ with differing next letters $a,b$ with, e.g., $a \prec b$ (the
other case being similar). $w_0$ is then defined as the smallest word
containg both occurrences of $wy$ and these letters $a,b$.

Now $x$ is quasiperiodic, hence some occurrence of $wyb$ must also
appear before some occurrence of $wya$, so we can find
between these two positions two occurrences of $wy$ with differing next
letters $c,d$ with $d \prec c$. We can then define $w_1$ similarly.
\end{proof}

\begin{theorem}\label{thm:1dcompdiffdegrees}
Let $A$ be a minimal one-dimensional subshift containing only strictly quasiperiodic points and
$x$ a point of $A$. Then for any Turing degree $d$ such that $\turdeg x\leq
d$, there exist a point $y\in A$ with Turing degree $d$.
\end{theorem}

\begin{proof}
To prove the theorem, we will give two computable functions $f: A\times
\cantor\to A$ and $g:A\to\cantor$ such that for any $x\in
A$ and $s\in\cantor$ we have $g(f(x,s))=s$. This means in terms of
Turing degrees:
\begin{equation}\label{inequality}
\turdeg s \leq \turdeg f(x,s) \leq \sup_T(\turdeg x,\turdeg s)
\end{equation}

That is to say, we give two algorithms, one ($f$) that given a point $x$ of $A$
and a sequence $s$ of $\cantor$ reversibly computes a point of $A$ that embeds
$s$, the second ($g$) retrieves $s$ from the computed point.

Let us now give $f$. Let $\prec$ be an order on $\Sigma_A$. Given a point $x\in
A$ and a sequence $s\in\cantor$, 
$f$ recursively constructs another point of $A$: it starts with a block
$C_{-1}=x_0$ and recursively constructs bigger and bigger blocks $C_i$ such that
$C_{i+1}$ has $C_i$ in its center. Furthermore these blocks are each centered
in $0$. So that the sequence $C_{-1}\to C_0 \to C_1 \to \dots \to C_i \to \dots$
converges to a point $c$ of $A$ having all $C_i$'s in its center. It is
sufficient to show then how $C_{i+1}$ is constructed from $C_{i}$.

$f$ works as follows: 
It searches for two consecutive occurrences of $C_i$ in $x$,
where the two first differing letters satisfy $a\prec b$ if $s_{i+1} =
0$ and $b \prec a$ if $s_{i+1} = 1$.
We know that $f$ will eventually succeed in finding these occurrences
due to Lemma ~\ref{lem:1dtwowords}.

Now we define $C_{i+1}$ as the word in $x$ where we find these two
occurrences, correctly cut so that the first occurrence of $C_i$ is
at its center, and its last letter is the differing letter of the
second occurrence. See Figure~\ref{fig:1dextwords}.

We thus have $f$, which is clearly computable. We give now $g$.

Given $C_i$ and $c$, one can compute $s_{i+1}$ easily: we just have to
look for the second occurrence of $C_i$ in $c$, the first one being in its
center. We then check whether the first differing letters between the blocks
following each occurrence are such that $e\prec f$ or $f\prec e$. This also
gives us $C_{i+1}$.

This means that from $c$, one can recover $s$. We know $C_{-1}=c_0$ and from
this information, we can get the rest: from $c$ and $C_i$, one computes easily
$C_{i+1}$ and $s_i$. We have constructed our function $g$.

So now if we take a sequence $s$ such that $\turdeg s > \turdeg x$, we can take
$y=c=f(x,s)$. It follows from inequality~\ref{inequality} that it has the same Turing degree as $s$
since $\turdeg s = \sup_T(\turdeg x,\turdeg s)$.

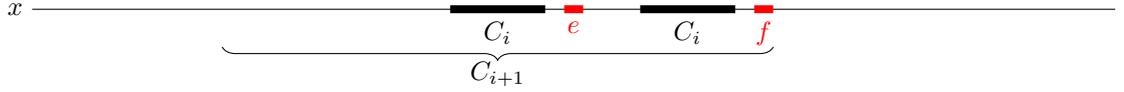
\begin{figure}
 \centering
 \begin{tikzpicture}[scale=0.25]
  \draw (-7,0) node[left] {$x$} -- (50,0);
\begin{scope}[shift={(0,0)}]
  \draw[line width=3pt] (15,0) -- +(5,0) node[midway,below] {$C_i$};
  \draw[line width=3pt] (25,0) -- +(5,0) node[midway,below] {$C_i$};
  \draw[line width=3pt,color=red] (21,0) -- +(1,0) node[midway,below] {$e$};
  \draw[line width=3pt,color=red] (31,0) -- +(1,0) node[midway,below] {$f$};
  \draw[decorate,decoration={brace,amplitude=5pt}] (32,-2) -- (3,-2)
  node[midway,below=2pt] {$C_{i+1}$};
\end{scope}
  
 \end{tikzpicture}
 \caption{\label{fig:1dextwords} How we construct $c_{i+1}$ from $c_i$. When
$s_{i+1}=0$, we have $e\prec f$ and $f\prec e$ otherwise. The words of
Lemma~\ref{lem:1dtwowords} are completed on the left with the block preceeding
them in $x$.
 }
\end{figure}

\end{proof}

\begin{corollary}
Every non-empty one-dimensional subshift $S$ containing only non computable
points has points
with different but comparable degrees.
\end{corollary}
\begin{proof}
	Take any minimal subshift of $S$. It must contain only strictly
	quasiperiodic points, so the previous theorem applies.
\end{proof}	

For effective subshifts, we can do better:
\begin{lemma}\label{lemma:zeroprime}
	Every non-empty one-dimensional effective subshift $S$ contains a minimal
	subshift $\tilde S$ whose language is of Turing degree less than or equal to $0'$.
\end{lemma}
$0'$ is the degree of the Halting problem.
\begin{proof}
	Let $\cal F$ be the computable set of forbidden patterns defining $S$.
	Let $w_n$ be a (computable) enumeration of all words.
	Define ${\cal F}_n$ as follows: ${\cal F}_{-1} = \emptyset$.
	Then if ${\cal F}_n \cup {\cal F} \cup \{w_{n+1}\}$ defines a
	non-empty subshift, then ${\cal F}_{n+1} = {\cal F} \cup
	\{w_{n+1}\}$ else ${\cal F}_{n+1} = {\cal F}_n$.
    
	Now take $\tilde {\cal F} = \cup_n {\cal F}_n$. It is clear from
	the construction that $\tilde {\cal F}$ is computable given the Halting
	problem. Moreover $\tilde {\cal F}$ defines a non-empty, minimal
	subshift $\tilde S$. More exactly the complement of $\tilde {\cal F}$ is
	exactly the set of patterns appearing in $\tilde S$.
\end{proof} 
This lemma cannot be improved: an effective subshift is built in
Ballier/Jeandel \cite{ballier2010} for which the language of every
minimal subshift is at least of Turing degree $0'$.

Now it is clear that any minimal subshift $\tilde S$ has a point computable
in its language, so that:
\begin{corollary}
	Every non-empty one-dimensional effective subshift with no computable point contains
configurations of every Turing degree above $0'$.
\end{corollary}	
We do not know if this can be improved. While it is true that
all minimal subshifts in \cite{ballier2010} have a language of
Turing degree at least $0'$, this does not mean that their
configurations have all Turing degree at least $0'$. In the construction of \cite{ballier2010},
there
indeed exist computable minimal points. The construction of Myers
\cite{Myers} has non-computable points, but points of low degree.

\subsection{Two-dimensional SFTs}
We now prove an analogous theorem for two dimensional SFTs.
We cannot use the previous result directly as it is not true  
that any strictly quasiperiodic configuration always contain a
strictly quasiperiodic (horizontal) line. Indeed, there exist
strictly quasiperiodic configurations, even in SFTs with no 
periodic configurations, where some line in the configuration is not
quasiperiodic (this is the case of the ``cross'' in Robinson's
construction \cite{Robinson}) or for which every line is periodic of
different period (such configurations happen in particular in the
Kari-Culik construction \cite{Culik, Kari14}).

We will first try to prove a result similar to
Lemma~\ref{lem:1dtwowords}, for which we will need an intermediate
definition and lemma.

\begin{definition}[line]
 A \emph{line} or \emph{$n$-line} of a two-dimensional configuration $x\in
\Sigma^{\ZZ^2}$ is a function $l:\ZZ\times H \to \Sigma$, with $H=h+\llbracket
0;n-1\rrbracket$, $h\in\ZZ$, such that 
$$x_{\mid \ZZ\times H}=l\textrm{.}$$
Where $n$ is the width of the line and $h$ the vertical placement.
\end{definition}

One can also define a line in a block by simply taking the intersection of both
domains. The notion of quasiperiodicity for lines is exactly the same as the one
for one dimensional subshifts. We need this notion for the following lemma,
that will help us prove the two-dimensional version of
Lemma~\ref{lem:1dtwowords}. We also think that this lemma might be of interest
in itself.

\begin{lemma}\label{lem:onlyquasiperlines}
 Let $A$ be a two-dimensional minimal subshift. There exists a point $x\in A$ such that all its
lines are quasiperiodic.
\end{lemma}
\begin{proof}
Let $\{(a_i,b_i)\}_{i \in \mathbb{N}}$ be an enumeration of
$\mathbb{Z} \times \mathbb{N}$ and $H_i = a_i + \llbracket 0;b_i\rrbracket$.

If $x$ is a configuration, denote by $p_i(x) : \mathbb{Z} \times H_i
\mapsto \Sigma$ the restriction of $x$ to $\mathbb{Z} \times H_i$.
We will often view $p_i$ as a map from $A$ to $(\Sigma^{H_i})^\mathbb{Z}$.
A \emph{horizontal} subshift is a subset of $\Sigma^{\mathbb{Z}^2}$
which is closed and invariant by a horizontal shift.

We will build by induction a non-empty horizontal
subshift $A_i$ of $A$ with the property that every configuration $x$ of
$A_i$ has the property that every line of support $H_j$, for any $j < i$, is
quasiperiodic.
More precisely, $p_j(A_i)$ will be a minimal subshift.

Define $A_{-1} = A$. If $A_i$ is defined, consider $p_{i+1}(A_i)$.
This is a non-empty subshift, so it contains a minimal subshift $X$.
Now we define the horizontal subshift $A_{i+1} = p_{i+1}^{-1} (X) \cap A_i$.
By construction $p_{i+1}(A_{i+1})$ is minimal.
Furthermore, for any $j<i$, $p_j(A_{i+1})$ is a non-empty subshift, and it is
included in
$p_j(A_j)$, which is minimal, hence it is minimal.

To end the proof, remark that by compactness $\cap_i A_i$ is non-empty,
as every finite intersection is non-empty.
\end{proof}

\begin{lemma}\label{lem:2dtwowords}
Let $A$ be a two-dimensional minimal subshift where all points
(equivalently, some point) have no horizontal period.

Let $x$ be a point of $A$ and $\prec$ be an order on $\Sigma_A$. For
each $n \in \mathbb{N}$, for any $n$-block $w$ extensible to $x$, there exist
two blocks $w_0$
and $w_1$, of the same size, both extensible to $x$ such that:
\begin{itemize}
 \item $w$ appears exactly twice in both $w_0$ and $w_1$, each on the $n$-line
of vertical placement 0.
 \item the first differing letters $e$ and $f$ in the blocks containing $w$ in their center are such
that $e\prec f$ in $w_0$ and $f\prec e$ in $w_1$.
\end{itemize} 
Here the word ``first'' refers to an adequate enumeration of
  $\mathbb{N} \times \mathbb{Z}$.
\end{lemma}
\begin{proof}
As the result is about patterns rather than configurations, and all points of a minimal subshift
have the same patterns, it is sufficient by Lemma~\ref{lem:onlyquasiperlines} to prove the result
when all lines of
$x$ are quasiperiodic.

Since $w$ appears in $x$, it appears a second time on the same $n$-line in $x$.
Since $x$ is not horizontally periodic, both occurrences
are in the center of different blocks. (The place where they differ
may be on a different line, though, if this particular $n$-line is periodic)

Now we use the same argument as lemma~\ref{lem:1dtwowords} on the
$m$-line containing both occurrences of $w$ and the first place they
differ. (Note that we cannot use directly the lemma as this $m$-line might
itself be periodic, but the proof still works in this case)
\end{proof}

\begin{theorem}\label{thm:2danycompdeg}
Let $A$ be a two-dimensional minimal subshift where all points
(equivalently, some point) have no horizontal period and let $x$ be a point of $A$.
Then for any Turing degree $d$
such that $\turdeg x\leq d$, there exists a point $y\in A$ with Turing degree
$d$.
\end{theorem}
\begin{proof}
 The proof is almost identical to the one of Theorem~\ref{thm:1dcompdiffdegrees},
 Lemma~\ref{lem:2dtwowords} being the two-dimensional counterpart of
Lemma~\ref{lem:1dtwowords}, the only difference being that we have to search simultaneously all
lines for the presence of two occurences of $C_i$ in order to construct $C_{i+1}$. One can see in
Figure~\ref{fig:2dextwords} how the $C_i$'s are contructed in this case.
\end{proof}

\begin{figure}
 \centering
 \begin{tikzpicture}[scale=0.1]
  \draw[very thick] (-5,-5) rectangle +(10,10) node[midway] {$C_i$};
  \draw (-9,-9) rectangle +(18,18);
  \draw[very thick] (20,-5) rectangle +(10,10) node[midway] {$C_i$};
  \draw (16,-9) rectangle +(18,18);
  \draw[very thick] (-34,-34) rectangle (34,34);
  \node at (-30,30) {$C_{i+1}$};
  \fill[color=red] (-5,-5) ++(14,12) rectangle +(-1,-1) node[left,below]
{$e$};
  \fill[color=red] (20,-5) ++(14,12) rectangle +(-1,-1) node[left,below]
{$f$};
  %\draw (-36,-5) -- (36,-5);
  %\draw (-36,5) -- (36,5);
  %\node at (-36,0) {$l$};
 \end{tikzpicture}
 \caption{\label{fig:2dextwords} How $C_{i+1}$ is constructed inductively from
$C_i$. $C_i$ is in the center of $C_{i+1}$. The letters $e$ and $f$ are the
first differing letters in the blocks
containing the $C_i$'s. Whether $e\prec f$ of $f\prec e$ depends on what symbol
we want to embed, 0 or 1.}
\end{figure}
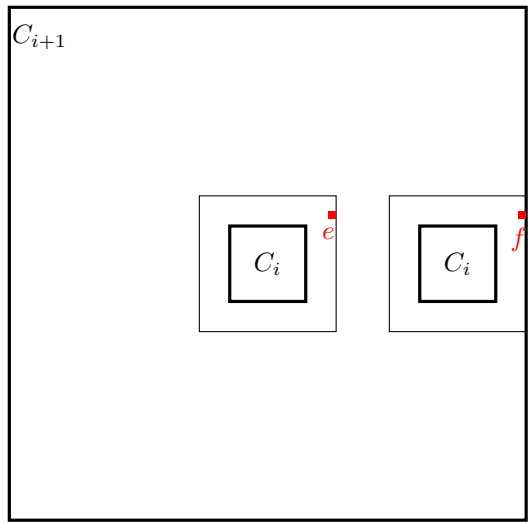

\begin{corollary}\label{cor:2dcompdiffdegrees}
	\label{cor:compdiffdegrees}
Every two-dimensional non-empty subshift $X$ containing only non-computable
points has points
with different but comparable Turing degrees.
\end{corollary}
\begin{proof}
$X$ contains a minimal subshift $A$, which cannot be periodic since it would otherwise contain
computable points. There are now two possibilities:
\begin{itemize}
 \item If $A$ contains a point with a horizontal period, then all points of $A$
have a horizontal period, and the result follows from Theorem~\ref{thm:1dcompdiffdegrees}, since
all points are strictly quasiperiodic in the vertical direction.
 \item Otherwise, it follows from Theorem~\ref{thm:2danycompdeg}.
\end{itemize}

\end{proof}	

Lemma \ref{lemma:zeroprime} is still valid in any dimensions so that
we have:
\begin{corollary}
	Every two-dimensional non-empty effective subshift (in particular
	any non-empty SFT) with no
	computable points contains points of any  Turing degree above $0'$.
\end{corollary}	

We conjecture that a stronger statement is true: The set of Turing degrees of
any subshift with no computable points is upward closed. To prove this,
it is sufficient to prove that for any subshift $S$ and any
configuration $x$ of $S$ (which is not minimal), there exists a
minimal configuration in $S$ of Turing degree less than or equal to the degree of $x$.
We however have no idea how to prove this, and no counterexample comes
to mind.
\bibliographystyle{elsarticle-harv}
\bibliography{biblio,books}
\end{document}